\documentclass{llncs}
\usepackage{amsmath}
\usepackage{amssymb}
\usepackage{comment}
\usepackage{IEEEtrantools}

\usepackage{graphics,graphicx}

\usepackage{algorithm}
\usepackage[noend]{algpseudocode}
\usepackage{IEEEtrantools}

\usepackage{todonotes}

\usepackage{color, soul}   
\title{Monotone Drawings of $k$-Inner Planar Graphs}  %(Versio 09 Archive)(21/8/2018)}
\author{Anargyros Oikonomou\inst{1}, Antonios~Symvonis\inst{2}}
\authorrunning{A.~Oikonomou\inst{1},A.~Symvonis\inst{2}}
\tocauthor{A.~Oikonomou\inst{1},A.~Symvonis\inst{2}}
\authorrunning{A.~Oikonomou\inst{1},A.~Symvonis\inst{2}}

\institute{
School of Electrical \& Computer Engineering \and
School of Applied Mathematical \& Physical Sciences\\ National Technical University of Athens, Greece
}
\begin{document}

\pagestyle{plain}
\maketitle

% ============================================================
\begin{abstract}
A \emph{$k$-inner planar graph} is a planar graph that has a plane drawing with at most $k$ \emph{internal vertices}, i.e., vertices that  do not lie on the boundary of the outer face of its drawing. An outerplanar graph is a $0$-inner planar graph.
In this paper, we show how to construct a monotone drawing of a  $k$-inner planar graph on a $2(k+1)n \times 2(k+1)n$ grid.
In the special case of an outerplanar graph,
we can produce a planar monotone drawing on a $n \times n$ grid,
improving  the results in~\cite{m:2,HR:2015}.
% This yields  monotone drawings for outerplanar graphs on a $2n \times 2n$  grid, improving  the results in~\cite{m:2,HR:2015}.
\end{abstract}
% ============================================================

\section{Introduction}
A \emph{straight-line drawing} $\Gamma$ of a graph G is a mapping of each vertex to a distinct point on the plane and of each edge to a straight-line segment between the vertices.
A path $P=\{p_0,p_1,\ldots,p_n\}$ is \emph{monotone} if there exists a line $l$ such that the projections of the vertices of $P$ on $l$ appear on $l$ in the same order as on $P$.
A straight-line drawing $\Gamma$ of a graph G is \emph{monotone}, if a \emph{monotone} path connects every pair of vertices.
We say that an angle $\theta$ is \emph{convex}
if $0 < \theta \leq \pi$. A tree $T$ is called \emph{ ordered} if the order of edges incident to any vertex is fixed.
We call a drawing $\Gamma$ of an  ordered tree $T$ rooted at $r$ \emph{near-convex monotone}
if it is monotone and any pair of consecutive edges incident to a vertex,
with the exception of a single pair of consecutive edges incident to  $r$,
form a convex angle.

\emph{Monotone graph drawing} has been lately a very active   research area and several interesting results have appeared
since its introduction by {Angelini et al.}~\cite{m:1}.
In the case of  trees, 
{Angelini et al.}~\cite{m:1} provided two algorithms that used ideas from number theory (Stern-Brocot trees~\cite{Brocot1860,Stern1858}~\cite[Sect. 4.5]{Graham:1994}) to produce  monotone tree drawings.
Their BFS-based algorithm  used an $O(n^{1.6}) \times O(n^{1.6})$ grid while their DFS-based algorithm used an $O(n) \times O(n^2)$ grid.
Later, {Kindermann et al.}~\cite{mt:1} provided an algorithm based on Farey sequence (see~\cite[Sect. 4.5]{Graham:1994}) that used an $O(n^{1.5}) \times O(n^{1.5})$ grid.
He and He in a series of papers~\cite{mt:3,HE:SIDMA,mt:2}  gave  algorithms, also based on Farey sequences, that eventually reduced the required grid size to  $O(n) \times O(n)$.
Their monotone tree drawing uses a $12n \times 12n$ grid which is asymptotically optimal as there exist trees which require at least ${n\over 9} \times {n \over 9}$ area~\cite{HE:SIDMA}.
In a recent paper,
Oikonomou and Symvonis~\cite{GDOS:1} followed a different approach from number theory
and gave an algorithm based on a simple  weighting  method  and  some  simple  facts from  geometry,
that draws a tree on an $n \times n$ grid.

Hossain and Rahman~\cite{HR:2015} showed that  by  modifying the embedding of a connected planar graph
we can  produce a planar monotone drawing of that graph on an $O(n) \times O(n^2)$ grid. To achieve that, they reinvented  the  notion of an \emph{orderly spanning tree} introduced by Chiang et al.~\cite{CLL:2005} and referred to it as  \emph{ good spanning tree}. 
Finally, He and He~\cite{HeHe2015:monot3ConPlanar} proved that Felsner's
algorithm~\cite{Felsner2001} builds convex monotone drawings of 3-connected planar graphs
on $O(n) \times O(n)$  grids in linear time.

%\begin{comment}
A \emph{$k$-inner planar graph} is a planar graph that has a plane drawing with at most $k$ \emph{inner vertices}, i.e., vertices that  do not lie on the boundary of the outer face of its drawing. An outerplanar graph is a $0$-inner planar graph. In this paper, we show how to construct a monotone drawing of a  $k$-inner planar graph on a $2(k+1)n \times 2(k+1)n$ grid. This yields  monotone drawings for outerplanar graphs on a $2n \times 2n$  grid, improving  the results in~\cite{m:2,HR:2015}. 
%\end{comment}
%Due to space limitations, we do not address the time complexity of our algorithm (it runs in linear time).
Proofs of lemmata/theorems can be found in the Appendix.

\section{Preliminaries}

\begin{comment}
Let $\Gamma$ be a drawing of a graph $G$ and $(u,v)$ be an edge from vertex $u$ to vertex $v$ in $G$.
The slope of edge $(u,v)$, denoted by $slope(u,v)$, is the angle spanned by a counter-clockwise rotation that brings a horizontal half-line starting at $u$ and directed towards increasing $x$-coordinates to coincide with the half-line starting at $u$ and passing through $v$.
%We consider slopes that are equivalent modulo $2\pi$ as the same slope.
%Observe that $slope(u,v)=slope(v,u) - \pi$.
\end{comment}

Let $T$ be a tree rooted at a vertex $r$. Denote by $T_u$ the sub-tree of $T$ rooted at vertex $u$. By $|T_u|$ we denote the number of vertices of $T_u$. In the rest of the paper, we assume that all tree edges are directed away from the root.

%NEW ADDITION=====================================================================================================
%Let $G_\phi$ be a planar embedding of a planar graph $G$ with $n$ total vertices.
%The embedding $G_\phi$ is called $k$-inner if there are at most $k$ vertices that do not belong to the outer face of $G_\phi$.
%=================================================================================================================

Our algorithm for obtaining a monotone plane drawing of a $k$-inner planar graph in based  on our ability: i) to produce for any rooted tree a compact  monotone drawing  that satisfies specific properties and ii) to identify for any planar graph a \emph{good spanning tree}.    

%\subsection{One Quadrant Monotone Tree Drawings}
%\input{algo.tex}

\begin{theorem}[Oikonomou, Symvonis~\cite{GDOS:1}]
\label{thm:GD17_oneQuadTree}
Given an  $n$-vertex tree $T$ rooted at vertex $r$, we can produce in $O(n)$ time a monotone  drawing  of $T$ where: 
i) the root $r$ is drawn at $(0,0)$, 
ii) the drawing in non-strictly slope-disjoint\footnote{See~\cite{GDOS:1} for the definition of \emph{non-strictly slope-disjoint} drawings of trees.}, 
iii) the drawing is contained in the first quadrant, and 
iv) it fits in an $n \times n$ grid. 
\end{theorem}

By utilizing and slightly modifying the algorithm that supports  Thm.~\ref{thm:GD17_oneQuadTree}, we can obtain a specific monotone tree drawing that we later use in our algorithm for monotone $k$-inner planar graphs.

\begin{theorem}
\label{thm:convexSecondOctTree}
Given an  $n$-vertex tree $T$ rooted at vertex $r$, we can produce in $O(n)$ time a monotone  drawing  of $T$ where: i) the root $r$ is drawn at $(0,0)$, ii) the drawing in non-strictly slope-disjoint, iii) the drawing is near-convex, iv)
the drawing is contained in the second octant (defined by two half-lines with a common origin and slope $\pi \over 4$ and $\pi \over 2$, resp.),   and v) it fits in a $2n \times 2n$ grid. 
%\marginpar{For a proof, see Appendix.}
\end{theorem}

The following definition of a \emph{good spanning tree} is due to Hossain and Rahman~\cite{HR:2015}.
Let $G$ be a connected embedded plane graph and let $r$ be a vertex of $G$ lying on the boundary of its outer face. Let $T$ be an ordered spanning tree of $G$ rooted in $r$ that respects the embedding of $G$.
Let $P(r, v) = \langle u_1(= r), u_2,..., u_k(= v) \rangle$ be the unique path in $T$ from the root $r$ to a vertex $v \neq r$.
The path $P(r, v)$ divides the children of $u_i$ $(1 \leq i < k)$, except $u_{i+1}$,
into two groups; the left group $L$ and the right group $R$.
A child $x$ of $u_i$ is in  group $L$ and, more specifically in a subgroup of $L$ denoted by $u_i^L$, if the edge $(u_i, x)$ appears before  edge $(u_i, u_{i+1})$ in
clockwise order of the edges incident to $u_i$ when the ordering starts from  edge $(u_i, u_{i-1})$.
Similarly, a child $x$ of $u_i$ is in  group $R$ and more specifically in a subgroup of $R$ denoted by $u_i^R$ if the edge $(u_i, x)$ appears after  edge $(u_i, u_{i+1})$ in clockwise order of the edges incident to $u_i$ when the ordering starts from  edge $(u_i, u_{i-1})$.
A tree $T$ is called a good spanning tree of $G$ if every vertex $v$ $(v \neq r)$ of $G$ satisfies the following two conditions with respect to $P(r, v)$ (see Figure~\ref{fig:goodTree} for an example  of a  good spanning tree (solid edges)):

\begin{enumerate}
\item[C1:] $G$ does not have a non-tree edge $(v,u_i), i < k$; and

\item[C2:] The edges of $G$ incident to  vertex $v$, excluding $(u_{k-1},v)$, can be partitioned into three disjoint (and possibly empty) sets $X_v,~Y_v$ and $Z_v$ satisfying the following conditions:
	\begin{enumerate}
	\item[a] Each of $X_v$ and $Z_v$ is a set of consecutive non-tree edges and $Y_v$ is a set of consecutive tree edges.
	\item[b] Edges of  $X_v, Y_v$ and $Z_v$ appear clockwise in this order after edge $(u_{k-1},v)$.
	\item[c] For each edge $(v,v')\in X_v$, $v' \in T_w$ for some $w \in {u_{i}^{L}}$, $i<k$, and for each edge $(v,v')\in Z_{v}$, $v' \in T_w$ for some $w \in {u_{i}^{R}}$, $i<k$.
	\end{enumerate}
\end{enumerate}

%An example of a  a good spanning tree (solid edges) is shown in Figure~\ref{fig:goodTree}.

\begin{figure}[t]
	\begin{minipage}[t]{0.40\textwidth}
		\centering
		\includegraphics[scale=0.4]{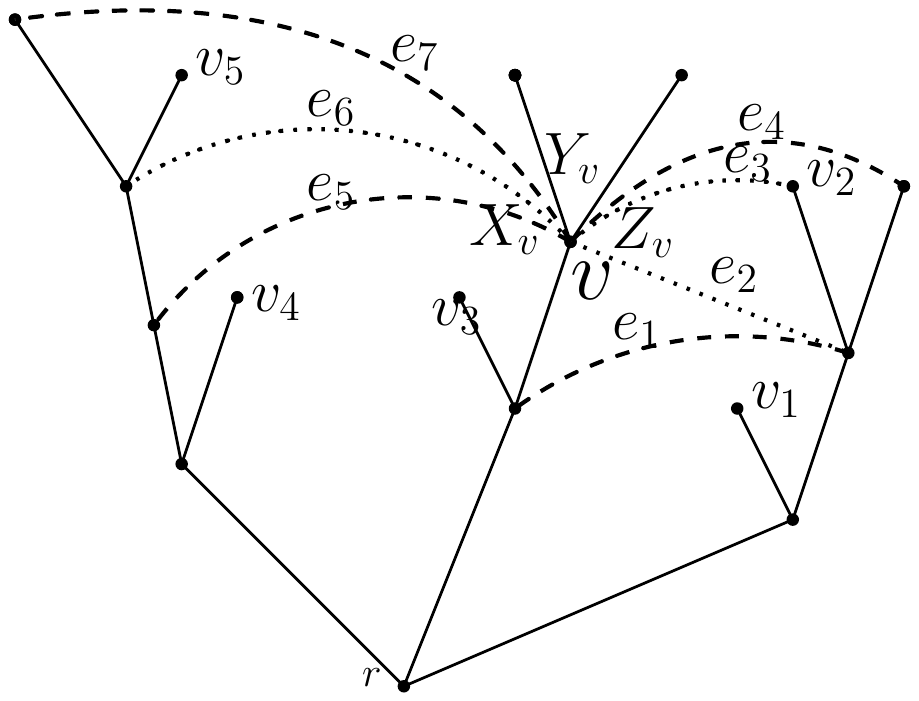}
		\caption{A good spanning tree.}
		\label{fig:goodTree}
	\end{minipage}
	\hfill
	\centering
	\begin{minipage}[t]{0.55\textwidth}
		\centering
		\includegraphics[scale=0.55]{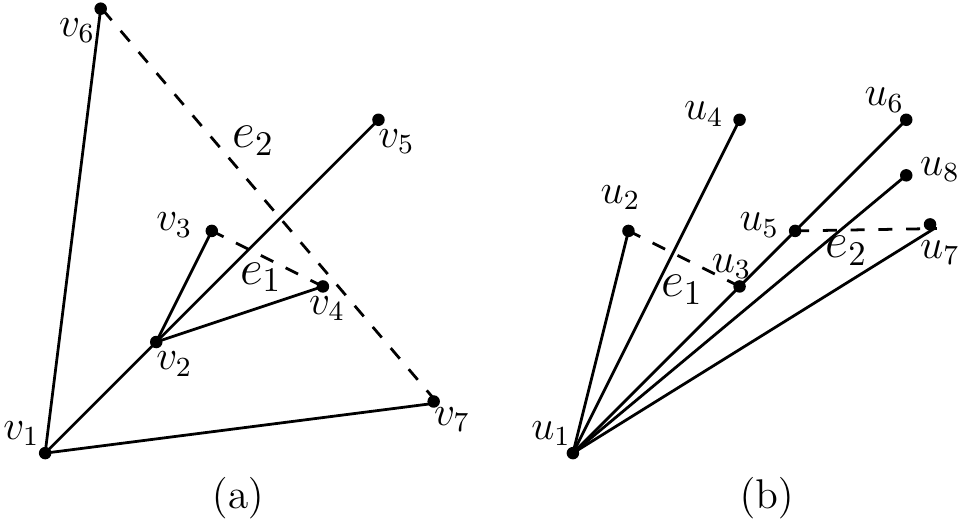}
		\caption{Dependencies between leader edges.}
		\label{fig:dep1_2}
	\end{minipage}
\end{figure}

%%\begin{figure}[t]
%	\centering
%		\includegraphics[scale=0.7]{figures/goodTree.pdf}
%		\caption{A plane graph and a good spanning tree (solid edges).}
%		\label{fig:goodTree}
%\end{figure}

\begin{theorem}[\cite{CLL:2005,HR:2015}]\label{theorem:good_trees}
Let $G$ be a connected planar graph of $n$ vertices.
Then $G$ has a planar embedding $G_\phi$ that contains a good spanning tree.
Furthermore, $G_\phi$ and a good spanning tree $T$ of $G_\phi$ can be found in $O(n)$ time.
\end{theorem}

Consider an embedded plane graph $G$ for which    a good spanning tree $T$ exists. 
We say that a non-tree edge $e$ of $G$ \emph{covers}  vertex $u$
if $u$ lies in the inner face delimited by the simple cycle formed by tree-edges of $T$ and $e$. Vertices on the cycle are not covered by edge $e$.

\begin{lemma}
\label{lem:innerVertices}
If a planar graph $G$ has a $k$-inner embedding,
then $G$ has a $k'$-inner embedding which contains a good spanning tree $T$, where $k' \leq k$.
Moreover, each non-tree edge covers  at most $k$ of $T$'s leaves. 
%\marginpar{For a proof, see Appendix.}
\end{lemma}

\begin{comment}
\begin{proof}
The algorithm that supports Theorem~\ref{theorem:good_trees} (see~\cite{HR:2015}), when given 
a $k$-inner embedding, it produces a planar embedding $G_\phi$ which contains a good spanning tree $T$
by changing the original embedding. It does so by  only moving \emph{$u$-components} and \emph{$(u,v)$-split components} out of a cycle induced by the constructed good spanning tree\footnote{For the definitions of a \emph{$u$-components} and \emph{$(u,v)$-split components} ans well as details of the good spanning tree construction algorithm see~\cite{HR:2015}.}.
Therefore, the number of vertices in the outer face of $G_\phi$ is equal or greater than the number of vertices in the original embedding of $G$,
which implies that $G_\phi$ is a $k'$-inner planar graph, where $k'\leq k$.
Moreover there are at most $k$ leaf vertices that are covered by any non-tree edge since there are at most $k$ inner vertices.
\qed
\end{proof}
\end{comment}

\section{$k$-inner Monotone Drawings}
The general idea for producing a monotone drawing  of plane $k$-inner graph $G$ which has a good spanning tree $T$ is to first obtain a monotone drawing of $T$ satisfying the properties of Thm~\ref{thm:convexSecondOctTree} and then to insert the remaining non-tree edges in a way that the drawing remains planar. The insertion of a non-tree edge may require to slightly adjust  the drawing obtained up to that point, resulting in a slightly larger drawing.
As it turns out, the insertion of a subset of the non-tree edges may violate the planarity of the drawing and, moreover, if these edges are considered in a proper order, the increase on the size of the drawing can be kept small (up to a factor of $k$ for each dimension).

Consider a plane graph $G$ that has a good spanning tree $T$. For any non-tree edge $e$ of $G$, we denote by $C(e)$ the set of leaf-vertices of $T$ covered by $e$. A non-tree edge $e$ is called a \emph{leader edge} if $C(e) \neq \emptyset$ and there doesn't exist another edge $e^\prime$ such that $C(e)=C(e^\prime)$ with $e^\prime$ lying in the inner of the cycle induced by the edges of $T$ and $e$. In Fig.~\ref{fig:goodTree}, leader edges are drawn as  dashed  edges. We also have that $C(e_1)=C(e_2)=C(e_3)=\left\{ v_1 \right\}$, $C(e_4)=\left\{ v_1, v_2 \right\}$, 
$C(e_5)=C(e_6)=\left\{ v_3, v_4 \right\}$, and $C(e_7)=\left\{ v_3, v_4, v_5 \right\}$.

\begin{lemma}
\label{lem:numberOfLeaders}
Let $G$ be a $k$-inner plane  graph that has a good spanning tree. Then, there exist at most $k$ leader  edges in $G$.
%\marginpar{For a proof, see Appendix.}
\end{lemma}

%\begin{comment}
\begin{proof} [Sketch]
Firstly observe that a boundary vertex of $G$ that is a leaf-vertex  in $T$ cannot be covered by any edge of $G$. That is, the set $C(e)$, for any non-tree edge $e$, contains only inner-vertices of $G$, and thus, $|C(e)| \leq k$.
The proof then easily follows by observing that for any two distinct leader edges $e_1, ~e_2$, exactly one of the following  statements holds: 
i) $C(e_1) \subset C(e_2)$,
ii) $C(e_2) \subset C(e_1)$,
iii) $C(e_1) \cap C(e_2) = \emptyset$. 
\qed
\end{proof}
%\end{comment}

\begin{comment}
In the following lemma, and in the rest of the paper, all inserted edges are drawn  as straight line segments.
\end{comment}

\begin{lemma}\label{lem:nonTreeOther}
Let $G$ be a plane graph that has a good spanning tree $T$ and
let $\Gamma_T$ be a monotone drawing of $T$ that is near-convex and non-strictly slope-disjoint.
Let $\Gamma_L$ be the  drawing produced if we add all leader  edges  of $G$ to $\Gamma_T$
and $\Gamma$  be the drawing produced if we add all non-tree edges  of $G$ to $\Gamma_T$.
Then,  $\Gamma_L$ is planar
if and only if  $\Gamma$ is planar.
\end{lemma}

Lemma~\ref{lem:nonTreeOther} indicates that we only have to adjust the original drawing of $T$ so that after the addition of all leader edges it is still planar. 
Then, the remaining non-tree edges can be drawn without violating planarity. 
Note that,  for the proof of Lemma~\ref{lem:nonTreeOther}, it is crucial  that the original drawing $\Gamma_T$ of $T$ is near-convex.

Consider the near-convex drawing of a good spanning tree $T$ rooted at $r$ that is non-strictly slope-disjoint. Assume that $T$ is drawn in the first quadrant with its root $r$ at $(0,0)$. 
Let $u\neq r$ be a vertex of $T$ and $p^u$ be its parent. Vertices $u$ and $p^u$ are drawn at grid points 
$(u_x, u_y)$ and $(p^u_x, p^u_y)$, resp. 
We define  the \emph{reference vector} of $u$ with respect to its parent in $T$, denoted by  
$\overrightarrow{\mathrm{u}}$, to be vector $\overrightarrow{\mathrm{u}}= (u_x - p^u_x, u_y - p^u_y)$. We emphasize that the reference vector $\overrightarrow{\mathrm{u}}$ of a tree vertex $u$ is defined wrt the original drawing of $T$ and does not change even if  the drawing of $T$ is  modified by our drawing algorithm.

The \emph{elongation of edge $(p^u,u)$ by a factor of $\lambda,~ \lambda \in \Bbb{N}^{+}$,} (also referred to as a \emph{$\lambda$-elongation})
translates the drawing of subtree $T_u$ along the direction of edge $(p^u,u)$ so that the new length of $(p^u,u)$ increases by  $\lambda$ times the  length of the reference vector $\overrightarrow{\mathrm{u}}$ of $u$. 
Since the elongation factor is a natural number, the new drawing is still a grid drawing. 
Let $\overrightarrow{\mathrm{u}}=(u^d_x, u^d_y)$.
After a  $\lambda$-elongation of $(p^u,u)$, vertex $u$ is drawn at point 
$u^{\prime} = (u_x +\lambda u^d_x, ~u_y+ \lambda u^d_y)$. A 0-elongation leaves the drawing unchanged. Note that, by appropriately selecting the elongation factor $\lambda$ for an upward tree edge $(p^u,u)$, we can reposition $u$ so that it is placed above any given point $z=(z_x, z_y)$.

If we insert the leader edges in the  drawing of the good spanning tree in an arbitrary order, we may have to adjust the drawing more than one time for each inserted edge. This is due to  dependencies between leader edges.  
Fig.~\ref{fig:dep1_2} describes the two types of possible dependencies. In the case of Fig.~\ref{fig:dep1_2}(a), the leader edge must by inserted first since $C(e_1)=\left\{ v_5 \right\} \subset \left\{v_3,v_4,v_5 \right\} = C(e_2)$.
Inserting leader $e_1$ so that it is not intersected by any tree edge can be achieved by elongating edges $(v_2,v_3)$ and 
$(v_2, v_4)$ by appropriate factors so that vertices $v_3$ and $v_4$ are both placed at grid points above (i.e., with larger $y$ coordinate) vertex $v_5$. 

In the case of Fig.~\ref{fig:dep1_2}(b), we have that $C(e_1)= \left\{ u_4 \right\}$,   
$C(e_2)= \left\{ u_8 \right\}$, and $C(e_1) \cap C(e_2) = \emptyset$. However, leader $e_1$ must be inserted first since one of its endpoints ($u_3$) is an ancestor of an endpoint ($u_5$) of $e_2$. Again, inserting leader $e_1$ so that it is not intersected by any tree edge can be achieved by elongating edges $(u_1, u_2)$ and 
$(u_1, u_3)$ by appropriate factors so that vertices $u_2$ and $u_3$ are both placed at grid points above vertex $u_4$.

\begin{lemma}
\label{lem:leaderOrdering}
Let $G$ be a plane graph that has a good spanning tree $T$ and
let $\Gamma_T$ be a drawing of $T$ that satisfies the properties of Thm~\ref{thm:convexSecondOctTree}. Then, there exists an ordering of the leader edges, such that if they are inserted into $\Gamma_T$ (with the appropriate elongations) in that order they need to be examined exactly once.
%\marginpar{For a proof, see Appendix.}
\end{lemma}

Our method for producing monotone drawings of $k$-inner planar graphs is summarized in Algorithm~\ref{algo:k_Inner}. A proof  that the produced drawing is actually a monotone plane drawing follows from the facts that i) there is always a good spanning tree with at most $k$ inner vertices (Lemma~\ref{lem:innerVertices}), ii) there is a monotone tree drawing satisfying the properties of Thm~\ref{thm:convexSecondOctTree}, iii) the operation of edge elongation on the vertices of a near-convex monotone non-strictly slope-disjoint tree drawing maintains these properties, iv) the ability to always insert the leader  edges into the drawing without violating planarity (through elongation), and v) the  ability to insert the remaining non-tree edges (Lemma~\ref{lem:nonTreeOther}).

\begin{algorithm}[t]
\caption{Monotone drawing of $k$-inner planar graphs}
\label{algo:k_Inner}
\begin{algorithmic}[1]
\Procedure{$k$-InnerPlanarMonotoneDrawing}{G}
\State \hspace*{-0.5cm}Input: An $n$-vertex planar graph $G$ and an embedding of  $G$ with $k$ inner vertices. 
\State \hspace*{-0.5cm}Output: A monotone planar drawing of $G$ on a   $2(k+1)n \times 2(k+1)n $ grid.
\State
	\State Let $G_\phi$ be an embedding of $G$ that has a good spanning tree $T$.
	\State $L \gets \text{List of the  \emph{leader}  edges of~} G, \text{~ordered wrt their dependencies}$ (Lemma~\ref{lem:leaderOrdering}).
	\State $\Gamma_T \gets \text{ Monotone drawing of~} T \text{~satisfying all properties of Thm~\ref{thm:convexSecondOctTree}}$.
	\State $\Gamma \gets \Gamma_T$
	\While{$L$ is not empty}
		\State Let $e=(u,v)$ be the next edge in $L$. Remove $e$ from $L$.
		\State Let $w$ be the vertex of $C(e)$ (drawn in $\Gamma$ at $(w_x, w_y)$) with the largest 
		\State  ~~~~~~$y$-coordinate.
		\State Let $p^u$ and $p^v$ be the parents of vertices $u$ and $v$ in $T$, resp.
		\State \emph{Elongate} Line edges $(p^u, u)$ and $(p^v,v)$ by appropriate factors so that $u$ and
		\State  ~~~~~~$v$ are placed in $\Gamma$ above vertex $w$.
		\State Insert edge $e$ into drawing $\Gamma$.
	\EndWhile
	\State Insert all remaining non-tree edges  into drawing $\Gamma$.
	\State Output $\Gamma$.
\EndProcedure
\end{algorithmic}
\end{algorithm}

Let $\overrightarrow{\mathrm{u}}=(u^d_x, u^d_y)$ and 
$\overrightarrow{\mathrm{v}}=(v^d_x, v^d_y)$ be the reference vectors of $u$ and $v$, resp.
When we process  leader edge $e=(u,v)$ (lines 9-15 of Algorithm~\ref{algo:k_Inner}),  factors
$\lambda_u=\max \left(0, \left\lceil  \frac{w_y-u_y}{u^d_y} \right\rceil \right)$ and
$\lambda_v=\max \left(0,  \left\lceil  \frac{w_y-v_y}{v^d_y} \right\rceil \right)$ are used for the elongation of edges 
$(p^u,u)$ and   $(p^v,v)$, resp. 
The use of these elongation factors ensures that both $u$ and $v$ are placed above vertex   $w$, and thus the insertion of edge $e$ leaves the drawing planar. When the leader edges are processed in the order dictated by their dependencies (line 6 of  Algorithm~\ref{algo:k_Inner}), we can show that:

\begin{comment}  %  OLD WRONG. REF VECTORS NOT NEEDED
\begin{lemma}
\label{lem:edgeElongationSubtreeSize}
Let $\Gamma$ be the partial drawing of $G$  immediately after the insertion of  leader edge $e=(u,v)$ in Algorithm~\ref{algo:k_Inner}. Let $\overrightarrow{\mathrm{u}}=(u^d_x, u^d_y)$ and 
$\overrightarrow{\mathrm{v}}=(v^d_x, v^d_y)$ be the reference vectors of $u$ and $v$, resp. Then, in $\Gamma$ the drawings of $T_u$ and $T_v$ are contained in grids of size
$(2n-u^d_x) \times (2n-u^d_y)$ and $(2n-v^d_x) \times (2n-v^d_y)$, resp.
\end{lemma}
\end{comment}

\begin{lemma}
\label{lem:edgeElongationSubtreeSize}
Let $\Gamma_T$ be the  drawing of the good spanning tree $T$  satisfying the properties of Thm~\ref{thm:convexSecondOctTree} and let  $\Gamma_{(u,v)}$ be the  drawing of $G$ immediately after the insertion of  leader edge $e=(u,v)$ by Algorithm~\ref{algo:k_Inner}. Let  $u$ and $v$ be drawn in $\Gamma_T$ at points $(u_x, u_y)$ and $(v_x,v_y)$, resp. Then, in $\Gamma_{(u,v)}$ the drawings of $T_u$ and $T_v$ are contained in $(2n-u_x) \times (2n-u_y)$ and $(2n-v_x) \times (2n-v_y)$ grids, resp.
\end{lemma}

Based on Lemma~\ref{lem:edgeElongationSubtreeSize}, we can easily show the main result of our paper.

\begin{theorem}
\label{thm:kInnerPlanar}
Let $G$ be an $n$-vertex $k$-inner planar graph. 
Algorithm~\ref{algo:k_Inner} produces a planar monotone drawing of $G$ on a
$2(k+1)n \times 2(k+1)n$
 grid.
\end{theorem}

\begin{comment}

\begin{corollary}
Let $G$ be an $n$-vertex outerplanar graph. 
Algorithm~\ref{algo:k_Inner} produces a planar monotone drawing of $G$ on a
$2n \times 2n$
 grid.
\end{corollary}
\end{comment}

A corollary of Thm~\ref{thm:kInnerPlanar} is that for   an $n$-vertex outerplanar graph $G$ 
Algorithm~\ref{algo:k_Inner} produces a planar monotone drawing of $G$ on a
$2n \times 2n$ grid. However, we can further reduce the grid size down to $n \times n$.

\begin{theorem}
Let $G$ be an $n$-vertex outerplanar graph. Then, there exists an $n \times n$  planar monotone grid drawing of $G$.
\end{theorem}

\begin{proof}
Simply observe that since an  outerplanar graph has no leader edges, the drawing produced by Algorithm~\ref{algo:k_Inner} is identical to that of the original drawing of the good spanning tree $T$. 
In Algorithm~\ref{algo:k_Inner} we used a $2n \times 2n$  drawing of $T$ in the second octant in order to simplify the elongation operation. 
Since outerplanar graphs have no leader edges,
they require no elongations and we can use instead the (first quadrant) $n \times n$ monotone tree drawing of~\cite{GDOS:1}, appropriately modified so that it yields a near-convex drawing.
\qed
\end{proof}

\section{Conclusion}
We defined the  class of $k$-inner planar graphs which bridges the gap between outerplanar and planar graphs.
For an $n$-vertex $k$-inner planar graph $G$, we provided  an algorithm that produces a $2(k+1)n \times 2(k+1)n$ monotone grid drawing of $G$.
%Note that the best known algorithm by Hossain and Rahman~\cite{HR:2015} uses an $O(n^2) \times O(n)$ grid irrespectively of whether  the graph is planar or outerplanar. 
Building algorithms for $k$-inner graphs that incorporate $k$ into their time complexity or into the quality of their solution is an interesting open problem.
%\vfill
%\vfill

\newpage
\bibliographystyle{splncs04}
\bibliography{bibl_09}

%\end{document}
\newpage
\appendix
\section*{Appendix}

\bigskip
\section{Additional material for Section~2}

\subsection{Some more definitions}
Let $\Gamma$ be a drawing of a graph $G$ and $(u,v)$ be an edge from vertex $u$ to vertex $v$ in $G$.
The slope of edge $(u,v)$, denoted by $slope(u,v)$, is the angle spanned by a counter-clockwise rotation that brings a horizontal half-line starting at $u$ and directed towards increasing $x$-coordinates to coincide with the half-line starting at $u$ and passing through $v$.

\bigskip

Angelini et al.~\cite{m:1} defined the notion of slope-disjoint tree drawings and proved that a every such drawing is monotone.  In order to simplify the presentation of our compact tree drawing algorithm presented in~\cite{GDOS:1}, we slightly extended the definition of \emph{slope-disjoint} tree drawings given by Angelini et al.~\cite{m:1}.   More specifically, we called a
 tree drawing $\Gamma$  of a rooted tree $T$  a \emph{non-strictly slope-disjoint} drawing if the following conditions hold:

\begin{enumerate}
\item For every vertex $u \in T$, there exist two angles $a_1(u)$ and $a_2(u)$, with 
$0 \boldsymbol{\leq} a_1(u)<a_2(u)\boldsymbol{\leq} \pi$ 
such that for every edge $e$ that is either in $T_u$ or enters $u$ from its parent, it holds that $a_1(u)<slope(e)<a_2(u)$.

\item For every two vertices $u, v \in T$ such that $v$ is a child of $u$, it holds that $a_1(u) \boldsymbol{\leq} a_1(v) < a_2(v) \boldsymbol{\leq} a_2(u)$.
  
\item For every two vertices $u_1, u_2$ with the same parent, it holds that either $a_1(u_1) < a_2(u_1) \boldsymbol{\leq} a_1(u_2) < a_2(u_2)$ or $a_1(u_2) < a_2(u_2) \boldsymbol{\leq} a_1(u_1) < a_2(u_1)$.
\end{enumerate}

The idea behind the original definition  of  slope-disjoint tree drawings is that all edges in the subtree $T_u$  as well as the edge entering $u$ from its parent will have slopes that \emph{strictly} fall within  the angle range 
$\left\langle a_1(u), a_2(u)\right\rangle$ defined for vertex $u$. $\left\langle a_1(u), a_2(u)\right\rangle$ is called the \emph{angle range of}  $u$ with   $a_1(u)$ and $a_2(u)$ being its \emph{boundaries}. In our extended definition, we allowed for angle ranges of adjacent vertices (parent-child relationship) or sibling vertices (children of the same parent)  to share angle range boundaries. Note that replacing the ``$\leq$'' symbols in our definition by the ``$<$'' symbol gives us the original definition of Angelini et al.~\cite{m:1} for the slope-disjoint tree drawings. Non-strictly slope-disjoint tree drawings are also monotone.

\bigskip
\subsection{Proof of Theorem~\ref{thm:convexSecondOctTree} }

\textbf{Theorem~2.}
\emph{Given an  $n$-vertex tree $T$ rooted at vertex $r$, we can produce in $O(n)$ time a monotone  drawing  of $T$ where: i) the root $r$ is drawn at $(0,0)$, ii) the drawing in non-strictly slope-disjoint, iii) the drawing is near-convex, iv)
the drawing is contained in the second octant (defined by two half-lines with a common origin and slope $\pi \over 4$ and $\pi \over 2$, resp.),   and v) it fits in a $2n \times 2n$ grid. }

\begin{proof}

We show how to extend/modify the algorithm presented in~\cite{GDOS:1} so that it satisfies all properties specified in the theorem. For brevity, we refer to the algorithm of~\cite{GDOS:1} as Algorithm \textsf{MQuadTD} (abbreviation for 
``\underline{M}onotone \underline{Quad}rant \underline{T}ree \underline{D}rawing'').

Algorithm  \textsf{MQuadTD} constructs a  non-strictly slope-disjoint tree drawing with the root of the tree drawn at $(0,0)$. So, the first two conditions are satisfied. In addition, the tree is drawn at the first quadrant of its root. We first show how Algorithm \textsf{MQuadTD} can be extended so that it produces near-convex drawings. Since Algorithm  \textsf{MQuadTD} produces a non-strictly 
slope-disjoint tree drawing, it splits   the angle-range of every internal tree vertex $u$ into as many slices as $u$'s children and assigns these slices to the children of  $u$. 

We denote by $p^u$ the parent of internal tree vertex $u$ . 
Note  that,  in a non-strictly slope-disjoint tree drawing that is drawn is an circular sector of size at most $\pi$, i.e., the angle-range of the root is $\left\langle 0,\alpha \right\rangle,~0<\alpha\leq \pi$, a  non-convex angle can be only formed between the edge $(p^u,u)$ entering an internal vertex $u$ and the edge connecting $u$ with its first or last child, say $v_l$ and $v_r$, resp. Moreover, in this case a convex angle may be formed only if the slope of edge $(p^u, u)$   falls within the angle-range of either $v_l$ or $v_r$. W.l.o.g., say that the slope of $(p^u, u)$  falls within the angle range of $v_l$; the other case is symmetric. It turns out that we can easily avoid creating a non-convex angle by simply drawing point $v_l$ on the first available grid point on the extension of $(p^u, u)$. 
If we place $v_l$ in this way, then the length of edge $(p^u, u)$  will be equal to that of $(u,v_l)$, which guarantees that this modification will not increase the required grid size.

Having modified Algorithm \textsf{MQuadTD} as described above, yields a drawing of $T$ that satisfies conditions i), ii) and iii). In order get a drawing that satisfies all five conditions we work as follows. We construct a new tree $T^\prime$ by getting two copies of $T$ and making them subtrees of a new root $r^\prime$. Since $T^\prime$ has $2n+1$ vertices, we can get a drawing of $T^\prime$ that lies in the first quadrant,  satisfies conditions i), ii) and iii) and, moreover, it fits on a $(2n+1) \times (2n+1)$
grid   with its root $r^\prime$ being the only vertex being place on the $x$ or the $y$ axis. By the construction of $T^\prime$ and its symmetry, we observe that its sub-drawing that lies in the second octant corresponds to  a drawing of the initial tree $T$. In this way, we obtain a drawing that satisfies all 5 conditions.
To complete the proof, simply note that all described modification of Algorithm \textsf{MQuadTD} can be easily accommodated in $O(n)$ time.
\qed
\end{proof}

The following Lemma  in not included in the main part of the paper.

\begin{lemma}
\label{lem:convexRegions}
Let $\Gamma_T$ be a drawing for  tree $T$ which satisfies all properties of Thm-\ref{thm:convexSecondOctTree} (and can be   produced by  the modification of algorithm \textsf{MQuadTD} as described in the proof of Thm-\ref{thm:convexSecondOctTree}). 
Then, if we substitute in  $\Gamma_T$ the edges entering the  leaves of $T$ by semi-infinite lines (i.e., rays), then $\Gamma_T$ partitions the first quadrant into unbounded convex regions. 
\end{lemma}

\begin{proof}
See Figure~\ref{fig:convPlane} for an example drawing. Note that since $\Gamma_T$ is monotone, it is also planar. 
Let $\Gamma^\prime_T $ be the drawing we obtain by substituting each leaf edge in $\Gamma_T $ by a ray. In addition, we add to $\Gamma^\prime_T $ two extra  rays, the first  coinciding with the positive $x$-axis and the second with the positive $y$-axis. 
The lemma states that all unbounded inner faces in $\Gamma^\prime_T$ are unbounded convex regions.

Since the original drawing of $T$ is near-convex and by the fact that the two extra rays lie on the $x$ and   $y$ axes,
all angles (in the first quadrant) between consecutive edges incident to any vertex are convex. Moreover, since the drawing of $T$ is also (non-strictly) slope-disjoint,
each pair of leaf edges is drawn completely inside non-overlapping (but possibly touching) angular sectors. Thus, 
the rays which substitute these edges diverge.  Then, the proof follows by the planarity of the original tree drawing  $\Gamma_T$ of $T$.
\qed 
\end{proof}

\begin{figure}[t]
	\begin{minipage}[t]{0.42\textwidth}
		\centering
		\includegraphics[scale=0.4]{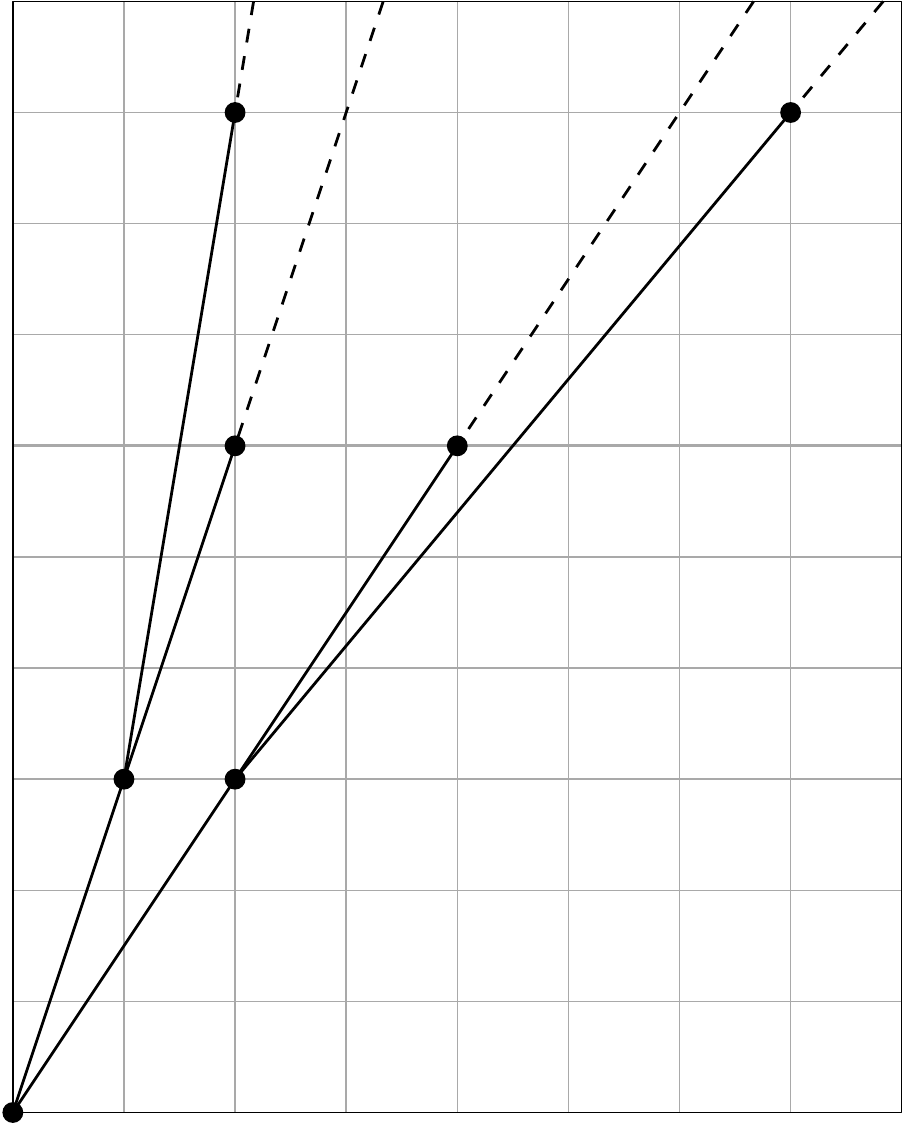}
		\caption{The unbounded convex regions of Lemma~\ref{lem:convexRegions}.}
		\label{fig:convPlane}
	\end{minipage}
	\hfill
	\centering
	\begin{minipage}[t]{0.54\textwidth}
		\centering
		\includegraphics[scale=0.45]{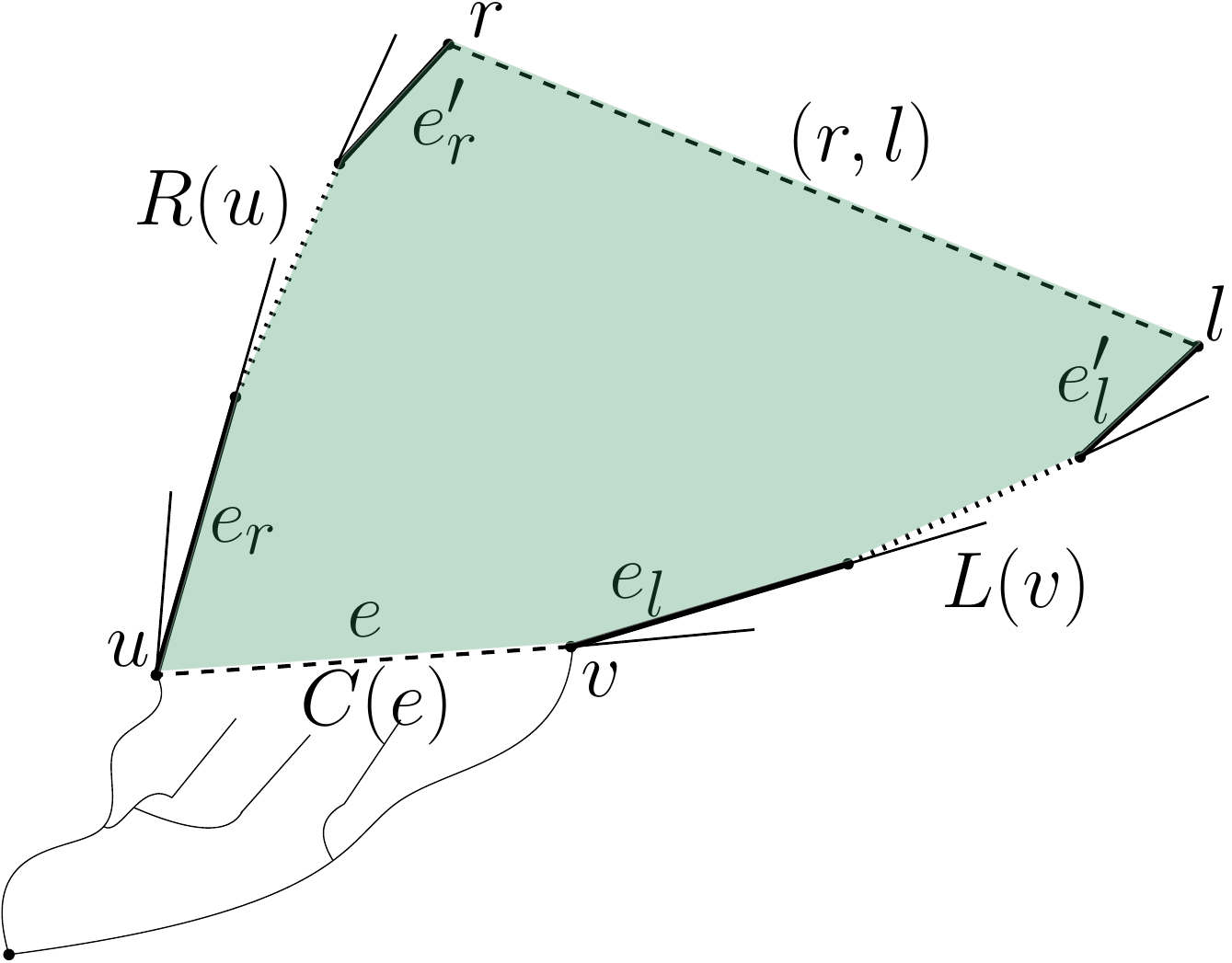}
		\caption{Illustration of Lemma~\ref{lem:convexLeaderRegion}.}
		\label{fig:conv}
	\end{minipage}
\end{figure}

\bigskip
\subsection{Proof of Lemma~\ref{lem:innerVertices}}  

\textbf{Lemma~1.}
\emph{If a planar graph $G$ has a $k$-inner embedding,
then $G$ has a $k'$-inner embedding which contains a good spanning tree, where $k' \leq k$.
Moreover, each non-tree edge covers  at most $k$ of $T$'s leaves.}

\begin{proof}
The algorithm that supports Theorem~\ref{theorem:good_trees} (see~\cite{HR:2015}), when given 
a $k$-inner embedding, it produces a planar embedding $G_\phi$ which contains a good spanning tree $T$
by changing the original embedding. It does so by  only moving \emph{$u$-components} and \emph{$(u,v)$-split components} out of a cycle induced by the constructed good spanning tree\footnote{For the definitions of a \emph{$u$-components} and \emph{$(u,v)$-split components} ans well as details of the good spanning tree construction algorithm see~\cite{HR:2015}.}.
Therefore, the number of vertices in the outer face of $G_\phi$ is equal or greater than the number of vertices in the original embedding of $G$,
which implies that $G_\phi$ is a $k'$-inner planar graph, where $k'\leq k$.
Moreover there are at most $k$ leaf vertices that are covered by any non-tree edge since there are at most $k$ inner vertices.
\qed
\end{proof}

\bigskip
\section{Additional material for Section~3}

Let $G$ be a plane graph and  $T$ being a good spanning tree of $G$.
For any vertex $u$ of $T$,
the \emph{leftmost path} (resp. \emph{rightmost path}) of $u$, denoted by $L(u)$ (resp. $R(u)$), 
starts at $u$, traverses  edges of $T$ by following the last (resp. first) tree edge of each vertex in the counter-clockwise order, and terminates at a leaf vertex. Recall that $C(e)$ denotes the set of leaf vertices of $T$ covered by then non-tree edge $e$. We will call  the non-tree vertices of $G$ that are not leader edges, \emph{ordinary}  edges. Finally, by $CH(P)$ we denote the convex hull of a point-set $P$.

The following technical lemma is not included in the main part of the paper.

\begin{lemma}\label{lem:convexLeaderRegion}
Let $G$ be a plane graph, $T$ be a good spanning tree of $G$, 
$\Gamma_T$ a  monotone drawing of $T$ which in near-convex and non-strictly slope-disjoint, and $e=(u,v)$ 
a leader edge of $G$. 
W.l.o.g., let the inner face of the cycle induced by $e$ and  edges of $T$ lie to the right of edge $e$ when we traverse it  from $u$ to $v$. 
Finally, let $N(e)$ be the set of vertices contained in paths $R(u)$ and $L(v)$.
Then, 
all vertices of  $N(e)$ are located on the the convex hull $CH(N(e))$.
Moreover, any ordinary  edge $e^\prime$ that covers exactly the same leaf vertices as $e$, (i.e.,  $C(e^\prime)=C(e)$) has one of its endpoints on $R(u)$ and the other on $L(v)$.
\end{lemma}

\begin{proof}
Refer to Figure~\ref{fig:conv}.
First, we prove that all vertices of  $N(e)$ are located on the convex hull $CH(N(e))$.
Let $r$ and $l$ be the leaves at the end of $R(u)$ and $L(v)$, resp., and  
also let $e_r$ (resp. $e'_r$) be the first (resp. last) edge of $R(u)$
and  $e_l$ (resp. $e'_l$) be the first (resp. last) edge of $L(v)$.

We  observe that if we substitute $e_r$ and $e_l$ with rays toward the decreasing $y$-coordinates,
then the two rays intersect.
Moreover, if we substitute $e'_r$ and $e'_l$ with rays toward the increasing $y$-coordinates,
then the two rays diverge. Both of these two facts are due to that  drawing $\Gamma_T$ is non-strictly slope-disjoint.
Combined with the fact that $\Gamma_T$ is near convex, proves that $CH(N(e))$ contains all vertices of $N(e)$.

Now we prove that every ordinary edge  $e^\prime$ that covers exactly the same leaf vertices as $e$  
(i.e., set  $C(e)$),
has one of its endpoints in $R(u)$ and the other in $R(v)$.
 
For a contradiction,
assume that there is an ordinary edge $e^\prime$ that covers exactly  set $C(e)$ of leaf vertices 
and does not have its endpoints in $N(e)$.
Edge $e^\prime$ cannot lie inside the cycle induced by $e$ and $T$ because it would violate the 
fact that $e$ is a leader edge.
Since $e^\prime$ do not have its endpoints in $N(e)$,
then $e^\prime$ covers at least one of leaf vertices $r$ or $l$, a contradiction. Thus, $e^\prime$ has its endpoints on $N(e)$. By taking also into account that $T$ is a good spanning tree, and thus, no non-tree edge  connects two vertices on the same path from the root to a leaf, we conclude that one of the endpoints of $e$ lies on $R(u$ and the other on $L(v)$.  
\qed
\end{proof}

\bigskip
\subsection{Proof of Lemma~\ref{lem:nonTreeOther}}

\textbf{Lemma~3.}
\emph{Let $G$ be a plane graph that has a good spanning tree $T$ and
let $\Gamma_T$ be a monotone drawing of $T$ that is near-convex and non-strictly slope-disjoint.
Let $\Gamma_L$ be the  drawing produced if we add all leader  edges  of $G$ to $\Gamma_T$
and $\Gamma$  be the drawing produced if we add all non-tree edges  of $G$ to $\Gamma_T$.
Then,  $\Gamma_L$ is planar
if and only if  $\Gamma$ is planar.}

\begin{proof}
First recall that when we add edges to an existing drawing we only add them as straight line segments.
\begin{description}
\item[``$\Leftarrow$'']
It trivially holds. $\Gamma_L$ is a subgraph of $\Gamma$ by definition. Any subgraph of a planar graph is planar.
\item[``$\Rightarrow$'']
We simply have to show that we can insert all ordinary edges into $\Gamma_L$ without violating planarity. Consider and arbitrary ordinary edge $e^\prime$. Then, there exists a leader edge $e=(u,v)$ such that $C( e^\prime)=C(e)$. If not, $e^\prime$ would be a leader edge. 
W.l.o.g., let the inner face of the cycle induced by $e$ and  edges of $T$ lie to the right of edge $e$ when we traverse it  from $u$ to $v$. 
Consider paths $R(u)$ and $L(v)$ and let $r$ and $l$ be their leaf endpoints, resp. Let $N(e)$ be the set of vertices of paths   $R(u)$ and $L(v)$.

Since  drawing $\Gamma_L$ is planar and, by Lemma\ref{lem:convexLeaderRegion},  we know that $CH(N(e))$ contains all points of $N(E)$,  the internal face induced by the cycle formed by edge $e$, $R(u)$, $L(v)$ and the line segment $(r,l)$ is empty and convex. Thus, since one endpoint of $e^\prime$  lies on $R(u)$ and the other on $L(v)$ (Lemma\ref{lem:convexLeaderRegion}), $e^\prime$ can be inserted to the drawing without violating planarity. In the same way, we can insert all ordinary edges. The proof is completed by observing that since $G$ is a plane graph, no two ordinary edges intersect.
\end{description}
\qed
\end{proof}

\bigskip
\subsection{Proof of Lemma~\ref{lem:leaderOrdering}}

\textbf{Lemma~4.}
\emph{Let $G$ be a plane graph that has a good spanning tree $T$ and
let $\Gamma_T$ be a drawing of $T$ that satisfies the properties of Thm~\ref{thm:convexSecondOctTree}. Then, there exists an ordering of the leader edges, such that if they are inserted into $\Gamma_T$ (with the appropriate elongations) in that order they need to be examined exactly once.}

\begin{proof}[Sketch]
We create a directed  \emph{dependency graph}  $D$ that has the leader edges as its vertices. The edges of $D$ correspond to the dependencies between leader edges of $G$. 
Let $e_1$ and $e_2$ be two leader edges of $G$. If $C(e_1) \subset C(e_2)$ we insert a directed edge from $e_1$ to $e_2$ in $D$. If $C(e_2) \subset C(e_1)$ we insert a directed edge from $e_2$ to $e_1$ in $D$. These dependency edges correspond to the dependency depicted in Figure~\ref{fig:dep1_2}(a). In addition, there may be a dependency between $e_1$ and $e_2$ even though $e_1 \cap e_2 = \emptyset$. This is depicted in Figure~\ref{fig:dep1_2}(b) and occurs when an endpoint of $e_1$ is an ancestor of $e_2$ or vice versa. In this case we insert an edge in $D$ from the ``ancestor''  to the ``descendant'' edge. The fact that $G$ is planar and $T$ is a good spanning tree imply that $D$ is actually a directed acyclic graph (DAG). Then, the ordering of the leader edges is obtained by topologically sorting DAG $D$.

We note that the dependency graph is not necessarily connected. Since we study $k$-inner planar graphs, there are at most $k$ leaders and, thus, $D$ has at most $k$ vertices, and consequently $O(k^2)$. So, we need $O(k^2)$ time for the topological ordering of $D$ which, together with the $O(n)$ time required for the recognition of the leader edges, yield an $O(n+k^2)$ time algorithm. However, we can improve it to $O(n+k)$ time by avoiding to insert transitive edges into $D$.
\qed
\end{proof}

\subsection{Proof of Lemma~\ref{lem:edgeElongationSubtreeSize}}

\textbf{Lemma~5.}
\emph{Let $\Gamma_T$ be the  drawing of the good spanning tree $T$  satisfying the properties of Thm~\ref{thm:convexSecondOctTree} and let  $\Gamma_{(u,v)}$ be the  drawing of $G$ immediately after the insertion of  leader edge $e=(u,v)$ by Algorithm~\ref{algo:k_Inner}. Let  $u$ and $v$ be drawn in $\Gamma_T$ at points $(u_x, u_y)$ and $(v_x,v_y)$, resp. Then, in $\Gamma_{(u,v)}$ the drawings of $T_u$ and $T_v$ are contained in $(2n-u_x) \times (2n-u_y)$ and $(2n-v_x) \times (2n-v_y)$ grids, resp.}
\begin{proof}
We only prove the lemma for $T_u$ since the case for $T_v$ is symmetric.

Let $p^u$ be the parent of $u$ in the good spanning tree $T$ used in Algorithm~\ref{algo:k_Inner}. 
The insertion of leader edge $(u,v)$ in the drawing causes the elongation of edge 
$(p^u,u)$ and, as a result, translates the drawing of $T_u$ along the direction of $(p^u,u)$. 

Observe that no edge in $T_u$ has  been elongated  in a previous step.
In order for an edge in $T_u$ to be elongated,
a leader edge $e'=(u',v')$,   such that $u'$ or $v'$ is in $T_u \backslash \left\{u \right\}$, must be selected for insertion.
So,  $u$ must be  an ancestor of $u'$ or $v'$. 
Then, leader $e^\prime$ depends on leader $e$ and has to  follow  $e$ in the sorted  leaders list $L$ (line 6 of Algorithm~\ref{algo:k_Inner}). Thus, no edge in $T_u$ has been elongated prior to the elongation of 
$(p^u,u)$.

The lemma follows from the facts that  i)~the initial drawing $\Gamma_T$ of the good spanning tree $T$ fits on a  $2n \times 2n$ grid  (Thm-\ref{thm:convexSecondOctTree}) and ii)~in $\Gamma_T$, vertex $u$ is placed at point $u_x, u_y)$. Therefore, $T_u$ is contained in a $(2n-u_x)\times (2n-u_y)$ grid.
\qed
\end{proof}

\bigskip
\subsection{Proof of Theorem~\ref{thm:kInnerPlanar}}

\textbf{Theorem~4.}
\emph{Let $G$ be an $n$-vertex $k$-inner planar graph. 
Algorithm~\ref{algo:k_Inner} produces a planar monotone drawing of $G$ on a
$2(k+1)n \times 2(k+1)n$
 grid.}

\begin{proof}
The proof that the  drawing produced by Algorithm~\ref{algo:k_Inner} is actually a monotone plane drawing follows from the facts that i) there is always a good spanning tree with at most $k$ inner vertices (Lemma~\ref{lem:innerVertices}), ii) there is a monotone tree drawing satisfying the properties of Thm~\ref{thm:convexSecondOctTree}, iii) the operation of edge elongation on the vertices of a near-convex monotone non-strictly slope-disjoint tree drawing maintains these properties, iv) the ability to always insert the leader  edges into the drawing without violating planarity (through elongation), and v) the  ability to insert the remaining ordinary edges (Lemma~\ref{lem:nonTreeOther}).

Now we turn our attention to the grid size of the drawing produced by Algorithm~\ref{algo:k_Inner}.
W.l.o.g., we assume that we deal with plane graphs that have only tree and leader edges. We can safely do so since, by Lemma~\ref{lem:nonTreeOther}, all ordinary edges can be inserted in the drawing without violating planarity.

We use induction on the number of leader edges.
For the base case where $k=0$, the grid bound holds by Theorem~\ref{thm:convexSecondOctTree}.
For the induction step,  we assume that the Theorem holds for plane graphs with less than $k$ leader edges, $k\geq 1$ and we show that it also holds for plane graphs with $k$ leader edges.

Firstly note that, since all edges in the drawing have slopes greater than $\frac{\pi}{4}$,
the grid side length is dictated by the maximum $y$-coordinate of a leaf vertex.

Consider a plane graph $G$ that  has a good spanning tree $T$  and $k$ of its non-tree edges are leader edges.  Let $e=(u,v)$ be the last leader edge considered by Algorithm~\ref{algo:k_Inner}. If we remove $e$ from $G$ we get a plane graph with $k-1$ leader edges.  

By the induction hypothesis, 
the produced drawing  $\Gamma$ of $G \setminus \left\lbrace e \right\rbrace$ fits   in a $2kn \times 2kn$ grid. Thus, the path from the root of $T$ to $u$ is contained  in a $2kn \times 2kn$ grid. 

Consider now the case where Algorithm~\ref{algo:k_Inner} inserts leader edge $e$ into the drawing $\Gamma$ it has already produced for graph $G \setminus \left\lbrace e \right\rbrace$. During the   iteration of the while-loop (lines 9-16 of Algorithm~\ref{algo:k_Inner}) for edge $e$,  at most two edges of $T$ are elongated. 
This means that the grid size of the drawing for $G$,  which is dictated as the maximum $y$ coordinate of a leaf vertex,
is the maximum between the $y$ coordinate of a leaf vertex in either $T_u$ or $T_v$, or it  does not change.
By Lemma~\ref{lem:edgeElongationSubtreeSize} and the fact that during the elongation of edge
$(p^u,u)$, vertex $u$ may be translated at most $u^d_y$ grid 
points\footnote{Recall that $\overrightarrow{\mathrm{u}}=(u^d_x, u^d_y)$ denotes  the reference vector of $u$.}
above the highest leaf of the drawing of  $\Gamma$ of $G \setminus \left\lbrace e \right\rbrace$, it follows that, after the insertion edge $e$, the path from the root of $T$ to $u$ fits in a
 $(2kn + u^d_y) \times (2kn + u^d_y)$  grid. So, the path from the root to any leaf of $T_u$ fits in a $(2kn + u^d_y + 2n - u_y) \times (2kn + u^d_y + 2n - u_y)$  grid. However, we have that:
 \begin{eqnarray*}
 2kn + u^d_y + 2n - u_y 	& = 	&  2(k+1)n - (u_y - u^d_y)\\
 							& \leq	&  2(k+1)n
 \end{eqnarray*}
 since $u_y \geq u^d_y$.
Thus, the  path from the root to any leaf of $T_u$ fits in a 
$2(k+1)n \times 2(k+1)n$  grid. Similarly, the same holds for the path from the root of $T$ to any leaf of $T_v$. This completes the proof. 
\qed
\end{proof}

\end{document}